\documentclass[11pt]{article}
\usepackage{amsmath}
\usepackage{amsfonts}
\usepackage{amsthm}
\usepackage{mathtools}
\usepackage{enumerate}
\usepackage{blkarray}
\usepackage[affil-sl]{authblk}
\numberwithin{equation}{section}

\theoremstyle{plain} \newtheorem{thm}{Theorem}[section]
\theoremstyle{plain} \newtheorem{define}{Definition}[section]
\theoremstyle{plain} \newtheorem{example}{Example}[section]
\theoremstyle{plain} \newtheorem{remark}{Remark}[section]
\theoremstyle{plain} \newtheorem{cor}{Corollary}[section]

\newcommand{\nn}{\nonumber} % For multline label suppression
\newcommand{\opname}{Walsh}
\newcommand{\opnamealt}{Hadamard}
\newcommand{\opnamecoeff}{\opname{} coefficients}

\newcommand{\wt}{\opname{} transform}

\newcommand{\wmat}{W}
\newcommand{\asymbol}{m}
\newcommand{\zsymbol}{z}
\newcommand{\mawt}[1]{\if#1a\wmat\else\if#10Z\else
    \errmessage{Bad matrix}\fi\fi}

\newcommand{\innerp}[2]{\left< {#1},{#2} \right>}
\newcommand{\na}[1]{t_{#1}}
\newcommand{\al}[1]{a_{#1}}

\newcommand{\xl}[2]{{#2}_{#1}}
\renewcommand{\int}[1]{\pi_{#1}}

\newcommand{\nl}{n}
\newcommand{\nlp}{n + 1}
\newcommand{\lo}{i}
\newcommand{\gtype}{\al{1}|\al{2}|\cdots|\al{\nl}}
\newcommand{\ngtype}{\na{1}\cdots\na{\nl}}
\newcommand{\zerogtype}{0|0|\ldots|0}

\newcommand{\xgtype}[1]{\xl{1}{#1}|\xl{2}{#1}|\cdots|\xl{\nl}{#1}}
\newcommand{\xngtype}[2]{\xl{1}{#1}|\xl{2}{#1}|\cdots|\xl{#2}{#1}}

\newcommand{\brokengtype}[2]{\xl{1}{#1}|\cdots|\xl{{#2}-1}{#1}
    |\xl{{#2}+1}{#1}|\cdots|\xl{\nl}{#1}}
\renewcommand{\vec}[1]{\;[#1]}
\newcommand{\sumogtypes}{\sum_{\gtype}}

\newcommand{\intvl}[2]{[\,#1, #2)}
\newcommand{\phenoscape}{fitness landscape}
\newcommand{\C}{\mathbb{C}}
\newcommand{\Csp}{\C^{\na{1},\na{2},\ldots,\na{\nl}}}

\newcommand{\sillyname}{}
\newcommand{\wronomega}{\sillyname{}\omega}
\newcommand{\wronmu}{\sillyname{}\mu}
\newcommand{\wrono}{\omega}
\newcommand{\wronm}{\mu}
\newcommand{\wronoa}[2]{\omega_{{#1},{#2}}}
\newcommand{\wronma}[2]{\mu_{{#1},{#2}}}
\newcommand{\wmatelt}{(-1)^{i_1j_1
    + i_2j_2 + \cdots + i_nj_n}}
\newcommand{\mawtaDef}[1]{
    \sum_{\xgtype{#1}} \prod_{\lo = 1}^{\nl}\;
    (1 - \na{\lo}\wronoa{\al{\lo}}{\xl{i}{b}})
    \vec{\xgtype{#1}}}
\newcommand{\mawtzDef}[1]{
    \sum_{\xgtype{#1}} \prod_{\lo = 1}^{\nl}\;
    \wronma{\al{\lo}}{\xl{i}{b}}
    \vec{\xgtype{#1}}}
\newcommand{\mawtaElt}{
    \prod_{\lo = 1}^{\nl}\;
    (1 - \na{\lo}\wronoa{\al{\lo}}{\xl{i}{b}})}
\newcommand{\mawtaElta}[3]{
    \prod_{\lo = 1}^{#3}\;
    (1 - \na{\lo}\wronoa{\xl{i}{#1}}{\xl{i}{#2}})}
\newcommand{\mawtzElt}{
    \prod_{\lo = 1}^{\nl}\;
    \wronma{\al{\lo}}{\xl{i}{b}}}
\newcommand{\mawtzElta}[3]{
    \prod_{\lo = 1}^{#3}\;
    \wronma{\xl{i}{#1}}{\xl{i}{#2}}}
\newcommand{\raisebf}[2]{{\vec{#1}}^{\textbf{\if%
    m#2w\else\if
    \asymbol#2w\else\if
    \zsymbol#2z\else\if%
    z#2z\else
    \errmessage{Bad superscript}\fi\fi\fi\fi}}}

\renewcommand{\deg}[1]{{\rm deg}(#1)}
\newcommand{\Additive}{Additive}
\newcommand{\additive}{additive}
\newcommand{\Span}{{\rm span}}

\newcommand{\final}{final}

\newenvironment{case}[1]%
{\noindent\it Case: {#1}.}
{\par\bigskip}

\author{Devin Greene}
\affil{American University, 4400 Massachusetts Avenue, NW, Washington, DC, United States, 20016}
\title{Multiallelic \wt{}s}
\date{}

\begin{document}

\maketitle

\begin{abstract}
\noindent
A closed formula multiallelic \opname{} (or \opnamealt{}) transform is
introduced.  Basic results are derived, and a statistical interpretation
of some of the resulting linear forms is discussed.
\end{abstract}

\section{Introduction}
The \opname{}, or \opnamealt{} matrix, is used in evolutionary biology
to detect and quantify epistasis, i.e.\ interactions, among the effects
of alleles over different loci.  However, applications using the
\wt{} are restricted to \emph{biallelic} systems: those
where each locus has two alleles.  Here we develop two versions of the
\opname{} matrix which described any multilocus, multiallele system.

Various authors have attempted to defined multiallelic analogs of the
\wt{}: See, for instance, \cite{faure2023extension},
\cite{metzger2023epistasis}.  In \cite{faure2023extension}, the authors
develop a recursive \wt{} which is closely related to the one we develop
here.  However, our approach uses closed formulas to describe the two
related versions of the transform. Furthermore, we emphasize the linear
forms which arise from these matrices, which generalize the notion of
\opnamecoeff{} introduced to evolutionary biology in
\cite{WEINREICH2013700}.

\section{Preliminaries}

Consider a genotype with $\nl$ loci, numbered 1 through $\nl$. Assume
that each locus $\lo$ has $\na{\lo} \ge 1$ possible alleles.  We label
the alleles of a particular locus $0, 1, 2,\ldots, \na{\lo} - 1$. If we
use the finite integer segment notation, i.e.

\begin{equation*}
\intvl{a}{b} = \{a, a + 1, \ldots, b - 1\},
\end{equation*}
where $a$, $b$ are integers,
then the set of alleles at the $\lo$th locus is $\intvl{0}{\na{\lo}}$.

We denote a genotype $g$ as follows.  Let $\al{\lo}$ be the allele at
the $\lo$th locus of $g$.  Then

\begin{equation*}
g = \gtype
\end{equation*}
Note that the set of genotypes is naturally equivalent to the
Cartesian product

\begin{equation*}
\intvl{0}{\na{1}} \times \cdots \times \intvl{0}{\na{n}}
\end{equation*}
This property will be fundamental for our sum-taking manipulations in
this paper.  For example, for any locus $i$, the summation
\begin{equation*}
\sumogtypes
\end{equation*}
can be replaced with
\begin{equation}\label{cartesianSum}
\sum_{\brokengtype{a}{i}} \sum_{a_i}
\end{equation}
Furthermore, if $f_{j,i}$ are terms in a product, where $f_{j,i}$
depends only on the locus $i$ and allele $j$, then
\begin{equation*}
\sumogtypes \prod_{i = 1}^\nl f_{a_i,i}
\end{equation*}
can be rewritten as
\begin{equation}\label{cartesianProd}
\sum_{\brokengtype{a}{k}}\prod_{0 \le i \ne k \le n}f_{a_i,i}
    \sum_{a_k}f_{a_k,k}
\end{equation}

The replacements \eqref{cartesianSum} and \eqref{cartesianProd} will be
used without comment for the remainder of this paper.

\bigskip

We stress that the labelling of alleles should not be interpreted to
imply any ranking or ordinality among those in the same locus.  We could
have labelled them with any set of symbols.  However, for the sake of
convenience in mathematical notation, we use integers.  Nevertheless, in
this paper's development, the allele in each locus labeled ``0'' will play a
special role as a ``base'' or ``anchor'' point.
Though we will not address it in this article, it can be shown that the
choice of an ``anchor'' allele at each locus does not in any fundamental
way treat such alleles as special or distinct from the others.

We define the set of \emph{\phenoscape{}s} to be the finite dimensional
Hilbert space $\Csp{}$ with the standard component-wise inner product,
where the components are indexed by the set of genotypes.  The standard
basis vector for genotype $\gtype{}$ is denoted using square brackets,
e.g.\ as $\vec{\gtype{}}$.

\section{Two dual multiallelic \wt{}s}

The \opname{}, or \opnamealt{} matrix, is a $2^{\nl}$ by $2^{\nl}$ matrix with
rows and columns indexed by bit strings of length $\nl$, whose
$i_1\cdots%
i_n, j_1\cdots%
j_n$th matrix element can be defined by the formula

\begin{equation*}
\wmatelt
\end{equation*}
The \opname{} matrix has been used to measure interactions among loci in
biallelic systems.  Here, we introduce two linear transformations which
generalize to an arbitrary number of alleles per locus.

\bigskip
Before presenting the definitions, we define two functions: the
$\wronomega$ function and the $\wronmu$ function.

The $\wronomega$ function is defined on pairs of elements in a set with
a unique zero element as

\begin{equation*}\label{wronOmega}
\omega_{x, y} =
\begin{cases}
1 & \text{if $x \ne 0$ and $x = y$}\\
0 & \text{otherwise}
\end{cases}
\end{equation*}
The $\wronmu$ function is defined on pairs of elements in a similar type
of set.
\begin{equation}\label{wronMu}
\mu_{x, y} =
\begin{cases}
1 & \text{if $x = 0$ or $y = 0$, otherwise}\\
-1 & \text{if $x = y$}\\
0 & \text{otherwise}
\end{cases}
\end{equation}
We remark that both of these functions are symmetric with respect to
their arguments, i.e.\ $\wronoa{x}{y} = \wronoa{y}{x}$ and $\wronma{x}{y} =
\wronma{y}{x}$.

\bigskip
In the following definition, recall that $\na{1},\ldots,\na{n}$ are the
number of alleles at loci 1,2,...,$n$.
\begin{define}\label{definition}
We define the linear transformations $\mawt{a}$ and $\mawt{0}$ in $\Csp$
with respect to the standard basis as follows.

\begin{equation}\label{mawtA}
\mawt{a} \vec{\gtype} = \mawtaDef{b}
\end{equation}
\begin{equation}\label{mawt0}
\mawt{0} \vec{\gtype} = \mawtzDef{b}
\end{equation}
\end{define}

\noindent
For convenience, we adopt the following superscript notation:
\begin{align*}
\raisebf{{\gtype}}{\asymbol} &= \mawt{a} \vec{\gtype}\\
\raisebf{{\gtype}}{\zsymbol} &= \mawt{0} \vec{\gtype}\\
\end{align*}

\begin{example}\label{example1stOrder}
To understand the motivation behind Definition \ref{definition},
consider the vectors $\raisebf{2|0}{\asymbol}$ and
$\raisebf{2|0}{\zsymbol}$ in $\C^{3,3}$.
From \eqref{mawtA} and \eqref{mawt0} above,

\begin{align}\label{mawtAExample}
\raisebf{2|0}{\asymbol} &=
\vec{0|0} + \vec{0|1} + \vec{0|2} + \vec{1|0} - 2\vec{2|0}\\
&+ \vec{1|1} + \vec{1|2} - 2\vec{2|1} - 2\vec{2|2}\nn
\end{align}
\begin{align}
\raisebf{2|0}{\zsymbol} &=
\vec{0|0} + \vec{0|1} + \vec{0|2} - \vec{2|0} - \vec{2|1} - \vec{2|2}
\label{mawt0Example}
\end{align}
If we take the inner product of a \phenoscape{} $v$ with each
vector above, we obtain the following linear forms.

\begin{equation}\label{mawtAExampleForm}
\innerp{v}{\raisebf{2|0}{\asymbol}} =
v_{00} + v_{01} + v_{02} + v_{10} - 2v_{20} + v_{11} + v_{12} - 2v_{21} - 2v_{22}
\end{equation}
\begin{equation}\label{mawt0ExampleForm}
\innerp{v}{\raisebf{2|0}{\zsymbol}} =
v_{00} + v_{01} + v_{02} - v_{20} - v_{21} - v_{22}
\end{equation}

We can rewrite \eqref{mawtAExampleForm}, up to a scalar multiple, as
$$
\frac{
(v_{20} - \frac{v_{00} + v_{10}}{2})
+ (v_{21} - \frac{v_{01} + v_{11}}{2})
+ (v_{22} - \frac{v_{02} + v_{12}}{2})
}{3}
$$
This is readily interpreted as being the average effect of replacing the
allele at the first locus with `2' as compared to the average of the
alternatives.

Rewriting \eqref{mawt0ExampleForm}, we obtain, up a scalar multiple,

\begin{equation*}
\frac{
(v_{20} - v_{00}) + (v_{21} - v_{01}) + (v_{22} - v_{02})
}{3}
\end{equation*}
which can be interpreted as the average effect of replacing `0' in the
first locus with `2'.
\end{example}

\begin{example}\label{example0thOrder}
The formulas \eqref{mawtA} and \eqref{mawt0} give the same vector at
\break$\vec{\zerogtype}$, indeed

\begin{equation*}
\raisebf{\zerogtype}{\asymbol} = \raisebf{\zerogtype}{\zsymbol}
= \sum_{\gtype} \vec{\gtype}
\end{equation*}
Thus if $v$ is a \phenoscape{} in $\Csp$, then

\begin{align*}
\frac{1}{\na{1}\cdots\na{\nl}}\innerp{v}{\raisebf{\zerogtype}{\asymbol}}
&= \frac{1}{\na{1}\cdots\na{\nl}}\innerp{v}{\raisebf{\zerogtype}{\zsymbol}}\\
&= \text{the average over all of $v$'s components}
\end{align*}
\end{example}

The following theorem constitutes the main result of this article.

\begin{thm}\label{mainTheorem}
\hspace{-1em}.

\begin{enumerate}
\item\label{main1}
$\mawt{a}$ and $\mawt{0}$ are self-adjoint.
\item\label{main2}

$\mawt{a}$ and $\mawt{0}$ are scaled inverses of each other.  Explicitly
\begin{equation}\label{scaledInverse}
\mawt{0}\mawt{a} = \mawt{a}\mawt{0} = \ngtype I
\end{equation}
\item\label{main3}

$\mawt{a}^2$ and $\mawt{0}^2$ are positive definite and block diagonal
with matching block subspaces. Explicitly

\begin{multline}
(\mawt{a}^2)_{\gtype, \xgtype{b}}\\ = \prod_{i=1}^\nl
\na{i}\left(\delta_{a_i,0}\delta_{b_i,0}
+ (1 - \delta_{a_i,0})(1 - \delta_{b_i,0})
    (\na{i}\delta_{a_i,b_i} - 1)\right)\hfil\label{waSquared}
\end{multline}
and
\begin{multline}
(\mawt{0}^2)_{\gtype, \xgtype{b}}\\ = \prod_{i = 1}^\nl
\left(\na{i} \delta_{a_i,0}\delta_{b_i,0}
+ (1 - \delta_{a_i,0})(1 - \delta_{b_i,0})(1 +
\delta_{a_i,b_i})\right)\hfil\label{w0Squared}
\end{multline}

\end{enumerate}
\end{thm}

\begin{proof}
We start with the matrix representations of $\mawt{a}$ and $\mawt{0}$
with respect to the standard basis.  (Strictly speaking, since
Definition \ref{definition} is a vector definition, we should
distinguish between left acting matrices and right acting matrices.
However, since $\wrono{}$ and $\wronm{}$ are symmetric on their
arguments, the distinction is irrelevant, and all the more so when we
have shown that the matrices are self-adjoint.)

\begin{equation*}
(\mawt{a})_{\gtype,\,\xgtype{b}}
= \mawtaElt
\end{equation*}
\begin{equation*}
(\mawt{0})_{\gtype,\,\xgtype{b}}
= \mawtzElt
\end{equation*}
Part \ref{main1} then follows immediately from the symmetry of $\wrono{}$
and $\wronm{}$ as was just alluded to.

We use induction for Part \ref{main2}.  The case $\nl = 0$ is trivial.
Assume validity at $\nl$, and consider the matrix element of the product
of $\mawt{a}$ and $\mawt{0}$ at $\nl + 1$:

\begin{equation}\label{WZProduct}
\sum_{\xngtype{c}{\nl + 1}}
    \mawtaElta{a}{c}{\nl + 1}\mawtzElta{c}{b}{\nl + 1}
\end{equation}
We write \eqref{WZProduct} as
\begin{equation*}
\sum_{\xgtype{c}}
\prod_{i = 1}^\nl (1 - \na{i}\wronoa{a_i}{c_i})\wronma{c_i}{b_i}
\sum_{c_{\nlp}}
(1 - \na{\nlp} \wronoa{a_{\nlp}}{c_{\nlp}})\wronma{c_{\nlp}}{b_{\nlp}}
\\
\end{equation*}
For the remainder of this part of the proof, the focus will be on the sum

\begin{equation}\label{PartSum}
\sum_{c_{\nlp}}
(1 - \na{\nlp} \wronoa{a_{\nlp}}{c_{\nlp}})\wronma{c_{\nlp}}{b_{\nlp}}
\end{equation}

\begin{case}{$a_{\nlp} = b_{\nlp} = 0$}
In this case
$\wronoa{a_{\nlp}}{c_{\nlp}} = 0$
and
$\wronma{c_{\nlp}}{b_{\nlp}} = 1$ for all
$c_{\nlp}$. Thus each summand in \eqref{PartSum} is 1, and so the total sum is
$\na{\nlp}$.  The desired result follows by induction.
\end{case}

\begin{case}{$a_{\nlp} = b_{\nlp} \ne 0$}
When $c_{\nlp} = 0$,
$\wronoa{a_{\nlp}}{c_{\nlp}} = 0$
and
$\wronma{c_{\nlp}}{b_{\nlp}} = 1$,
in which case the summand is 1.  When
$c_{\nlp} = a_{\nlp} = b_{\nlp}$,
$\wronoa{a_{\nlp}}{c_{\nlp}} = 1$
and
$\wronma{c_{\nlp}}{b_{\nlp}} = -1$, thus the summand is $\na{\nlp} - 1$.
For other values of $c_{\nlp}$,
$\wronma{c_{\nlp}}{b_{\nlp}} = 0$,
thus the summand is 0. It follows that the sum \eqref{PartSum} is
$\na{\nlp}$. The result follows by induction.
\end{case}

\begin{case}{$a_{\nlp} \ne b_{\nlp} = 0$}
In this case
$\wronma{c_{\nlp}}{b_{\nlp}} = 1$
for all $c_{\nlp}$, and so the sum \eqref{PartSum} can be expressed as
$\sum_{c_{\nlp}} (1 - \na{\nlp}\wronoa{a_{\nlp}}{c_{\nlp}})$.
If $c_{\nlp} \ne a_{\nlp}$, then
$\wronoa{c_{\nlp}}{a_{\nlp}} = 0$, in which case the summand is 1.
If
$c_{\nlp} = a_{\nlp}$, then
$\wronoa{c_{\nlp}}{a_{\nlp}} = 1$,
and the summand is $1 - \na{\nlp}$.
The total sum is then 0, as it should be since
$\xngtype{a}{\nlp} \ne \xngtype{b}{\nlp}$
\end{case}

\begin{case}{$b_{\nlp} \ne a_{\nlp} = 0$}
In this case
$\wronoa{c_{\nlp}}{a_{\nlp}} = 0$
for all $c_{\nlp}$, so \eqref{PartSum} can be expressed as
$\sum_{c_{\nlp} = 0}^{\na{\nlp} - 1} \wronma{c_{\nlp}}{b_{\nlp}}$.
$\wronma{c_{\nlp}}{b_{\nlp}} = 1$ when $c_{\nlp} = 0$, and
$\wronma{c_{\nlp}}{b_{\nlp}} = -1$ when $c_{\nlp} = b_{\nlp}$.
Otherwise, $\wronma{c_{\nlp}}{b_{\nlp}} = 0$.
Thus the sum \eqref{PartSum} is 0, as it should be since
$\xngtype{a}{\nlp} \ne \xngtype{b}{\nlp}$
\end{case}

\begin{case}{$a_{\nlp} \ne b_{\nlp}$, and $a_{\nlp},b_{\nlp} \ne 0$}
In this case when $c_{\nlp} = 0$,
$\wronoa{a_{\nlp}}{c_{\nlp}} = 0$
and $\wronma{c_{\nlp}}{b_{\nlp}} = 1$
thus the summand is 1.
When $c_{\nlp} = a_{\nlp}$
then $\wronoa{a_{\nlp}}{c_{\nlp}} = 1$ and
$\wronma{c_{\nlp}}{b_{\nlp}} = 0$,
thus the summand is 0.
When $c_{\nlp} = b_{\nlp}$
then $\wronoa{a_{\nlp}}{c_{\nlp}} = 0$ and
$\wronma{c_{\nlp}}{b_{\nlp}} = -1$
thus the summand is $-1$.
Thus the sum \eqref{PartSum} totals to 0, as it should since
$\xngtype{a}{\nlp} \ne \xngtype{b}{\nlp}$
\end{case}
\noindent
This completes the proof for Part \ref{main2}.

\bigskip

For Part \ref{main3}, $\mawt{a}^2$ is positive definite since $\mawt{a}$
is invertible, and $\mawt{a}\mawt{a}^* = \mawt{a}^2$.  Similarly for
$\mawt{0}^2$.  The claim respecting block diagonality follows from the
formulas \eqref{waSquared} and \eqref{w0Squared}, which we now prove by
induction on $\nl$.  For both $\mawt{0}$ and $\mawt{a}$, the case where
$\nl
= 0$ is trivial.  Assume validity at $\nl$, and consider a matrix
element of $\mawt{a}$ at $\nlp$:

\begin{multline}\label{inductionW}
(\mawt{a}^2)_{\xngtype{a}{\nlp},\xngtype{b}{\nlp}}\\
= \sum_{\xngtype{c}{\nlp}}
    \prod_{i=1}^{\nlp}(1 - \na{i}\wronoa{a_i}{c_i})
    (1 - \na{i}\wronoa{c_i}{b_i})\hfil\\
= \sum_{\xgtype{c}}
    \prod_{i = 1}^{\nl}(1 - \na{i}\wronoa{a_i}{c_i})
    (1 - \na{i}\wronoa{c_i}{b_i})\hfil\\
    \times \sum_{c_{\nlp}}
    (1 - \na{\nlp}\wronoa{a_{\nlp}}{c_{\nlp}})
    (1 - \na{\nlp}\wronoa{c_{\nlp}}{a_{\nlp}})\hfil
\end{multline}
We focus on the term

\begin{equation}\label{WTerm}
\sum_{c_{\nlp}}
    (1 - \na{\nlp}\wronoa{a_{\nlp}}{c_{\nlp}})
    (1 - \na{\nlp}\wronoa{c_{\nlp}}{b_{\nlp}})
\end{equation}
By induction, in each case that follows, it suffices to show
that \eqref{WTerm} is equivalent to the expression

\begin{equation}\label{targetW}
\na{\nlp}\left(\delta_{a_i,0}\delta_{b_i,0}
+ (1 - \delta_{a_i,0})(1 - \delta_{b_i,0})
    (\na{\nlp}\delta_{a_i,b_i} - 1)\right)
\end{equation}

\begin{case}{$a_i = b_i = 0$}
In this case \eqref{WTerm} is $\na{\nlp}$, as is \eqref{targetW}.
\end{case}
\begin{case}{$a_i = b_i \ne 0$}
In this case \eqref{WTerm} is
\begin{equation*}
(1 - \na{\nlp})^2 + \na{\nlp} - 1 = (\na{\nlp} - 1)\na{\nlp}
\end{equation*}
as is \eqref{targetW}
\end{case}
\begin{case}{$a_i \ne b_i = 0$ or $b_i \ne a_i = 0$}
In this case \eqref{WTerm} is
\begin{equation*}
(1 - \na{\nlp}) + (\na{\nlp} - 1) = 0
\end{equation*}
as is \eqref{targetW}
\end{case}
\begin{case}{$a_i \ne b_i$, $a_i, b_i \ne 0$}
In this case \eqref{WTerm} is
\begin{equation*}
2(1 - \na{\nlp}) + (\na{\nlp} - 2) = -\na{\nlp}
\end{equation*}
as is \eqref{targetW}
\end{case}

\noindent
Now consider the matrix element of $\mawt{0}$ at $\nlp$
\begin{multline}\label{inductionA}
(\mawt{0}^2)_{\xngtype{a}{\nlp},\xngtype{b}{\nlp}}\\
= \sum_{\xngtype{c}{\nlp}}
    \prod_{i=1}^{\nlp}\wronma{a_i}{c_i}
    \wronma{c_i}{b_i}\hfil\\
= \sum_{\xgtype{c}}
    \prod_{i = 1}^{\nl}\wronma{a_i}{c_i}
    \wronma{c_i}{b_i}\hfil\\
    \times \sum_{c_{\nlp}}
    \wronma{a_{\nlp}}{c_{\nlp}}
    \wronma{c_{\nlp}}{a_{\nlp}}\hfil
\end{multline}
Similarly to how we handled $\mawt{a}$, we focus on the term

\begin{equation}\label{Aterm}
\sum_{c_{\nlp}}
    \wronma{a_{\nlp}}{c_{\nlp}}
    \wronma{c_{\nlp}}{a_{\nlp}}
\end{equation}
and show by induction that each case has \eqref{Aterm} equal to
\begin{equation}\label{targetA}
\na{i} \delta_{a_i,0}\delta_{b_i,0}
+ (1 - \delta_{a_i,0})(1 - \delta_{b_i,0})(1 +
\delta_{a_i,b_i})
\end{equation}

\begin{case}{$a_i = b_i = 0$}
In this case \eqref{Aterm} is $\na{\nlp}$, as is
\eqref{targetA}.
\end{case}
\begin{case}{$a_i = b_i \ne 0$}
In this case \eqref{Aterm} is 2, as is \eqref{targetA}.
\end{case}
\begin{case}{$a_i \ne b_i = 0$ or $b_i \ne a_i = 0$}
In this case \eqref{Aterm} is 0, as is \eqref{targetA}.
\end{case}
\begin{case}{$a_i \ne b_i$, $a_i, b_i \ne 0$}
In this case \eqref{Aterm} is 1, as is \eqref{targetA}.
\end{case}

\end{proof}

For the remainder of this paper, we will identify $\mawt{a}$ and
$\mawt{0}$ with their matrix representations on the standard basis.

\begin{cor}\label{coeffFinder}
The sets

\begin{equation*}
\{\raisebf{\gtype}{\asymbol}\}_{\gtype}\\
\end{equation*}
and
\begin{equation*}
\{\raisebf{\gtype}{\zsymbol}\}_{\gtype}
\end{equation*}
are bases for $\Csp$.
Furthermore, if $v$ is a \phenoscape{}, and the coefficients
in the first basis are $c_{\gtype}$, i.e.\

\begin{equation*}
v = \sum_{\gtype} c_{\gtype} \raisebf{\gtype}{m},
\end{equation*}
then

\begin{equation*}
c_{\gtype} = \frac{1}{\na{1}\cdots\na{\nl}}\sum_{\xgtype{b}}
\mawt{0}_{\gtype, \xgtype{b}} v_{\xgtype{b}}
\end{equation*}

If $s_{\gtype}$ are the coefficients in the second basis, i.e.\

\begin{equation*}
v = \sum_{\gtype} s_{\gtype} \raisebf{\gtype}{z}
\end{equation*}
then

\begin{equation*}
s_{\gtype} = \frac{1}{\na{1}\cdots\na{\nl}}\sum_{\xgtype{b}}
\mawt{a}_{\gtype, \xgtype{b}} v_{\xgtype{b}}
\end{equation*}
\end{cor}

\begin{proof}
This follows immediately from Theorem \ref{mainTheorem} and basis
changing methods in linear algebra.
\end{proof}

\begin{define}
The \emph{degree} of a genotype $\gtype$ is the number of non-zero
alleles, i.e.\

\begin{equation*}
\deg{\gtype} = \#\{i|\al{i} \ne 0\}
\end{equation*}
\end{define}

\begin{cor}\label{degOrtho}
Let $\gtype$ and $\xgtype{b}$ be two genotypes with different degrees.
Then

\begin{equation*}
\innerp{\raisebf{\gtype}{m}}{\raisebf{\xgtype{b}}{m}} = 0
\end{equation*}
and
\begin{equation*}
\innerp{\raisebf{\gtype}{z}}{\raisebf{\xgtype{b}}{z}} = 0
\end{equation*}
\end{cor}

\begin{proof}
If $\gtype$ and $\xgtype{b}$ have different degree, then there must a
locus $i$ such that one of $a_i, b_i$ is zero the the other non-zero.
Then, by \eqref{waSquared} and \eqref{w0Squared},
\begin{multline*}
\innerp{\raisebf{\gtype}{m}}{\raisebf{\xgtype{b}}{m}}\\
= \innerp{\mawt{a}^2\vec{\gtype}}{\vec{\xgtype{b}}}\hfil\\
= \prod_{i=1}^\nl
\na{i}\left(\delta_{a_i,0}\delta_{b_i,0}
+ (1 - \delta_{a_i,0})(1 - \delta_{b_i,0})
    (\na{i}\delta_{a_i,b_i} - 1)\right) = 0\hfil
\end{multline*}
and
\begin{multline*}
\innerp{\raisebf{\gtype}{z}}{\raisebf{\xgtype{b}}{z}}\\
 = \innerp{\mawt{0}^2\vec{\gtype}}{\vec{\xgtype{b}}}\hfil\\
 = \prod_{i = 1}^\nl
\left(\na{i} \delta_{a_i,0}\delta_{b_i,0}
+ (1 - \delta_{a_i,0})(1 - \delta_{b_i,0})(1 +
\delta_{a_i,b_i})\right) = 0\hfil
\end{multline*}

\end{proof}

\begin{cor}\label{orthoDecomp}
Define the following subspaces:
\begin{equation*}
S_{i,\zsymbol} = \Span(\{\raisebf{\gtype}{z}|\deg{\gtype} = i\})
\end{equation*}
and
\begin{equation*}
S_{i,\asymbol} = \Span(\{\raisebf{\gtype}{m}|\deg{\gtype} = i\})
\end{equation*}
Then

\begin{enumerate}
\item
$S_{i,\zsymbol} = S_{i,\asymbol}$ for each $i$, and we will denote the
common subspace $S_i$, and
\item
$\Csp = S_0 \oplus S_1 \oplus \cdots \oplus S_n$ is an orthogonal
decomposition.
\end{enumerate}
\end{cor}

\begin{proof}
That
\begin{equation*}
\Csp = \bigoplus_{i=1}^\nl S_{i,\zsymbol} = \bigoplus_{i=1}^\nl S_{i,\asymbol}
\end{equation*}
follows immediately from Corollary \ref{degOrtho}.

Let $\gtype$ and $\xgtype{b}$ be two genotypes of different degree.
Taking inner products, we have
\begin{align*}\label{ZandZ^2}
&\innerp{\mawt{0} \vec{\gtype}}{\mawt{a} \vec{\xgtype{b}}}\\
&= \innerp{\mawt{a}\mawt{0} \vec{\gtype}}{\vec{\xgtype{b}}}\\
&= \na{1}\cdots\na{\nl}\innerp{\vec{\gtype}}{\vec{\xgtype{b}}} = 0
\end{align*}
This shows that $S_{i,\asymbol} = S_{i,\zsymbol}$ for each $i$,
completing the proof.
\end{proof}

\begin{remark}
In the presentation of the bases derived from $\mawt{a}$ and $\mawt{0}$
(i.e.\ the columns of their matrices when written in the standard
basis), we have avoided any discussion of the appropriate scaling of
these vectors or that of the coefficients corresponding to them,
preferring to express these vectors in integral form.  The magnitude
of any vector in the $\mawt{a}$- or $\mawt{0}$-basis can be computed
directly or by using one of the formulas \eqref{waSquared} or
\eqref{w0Squared}.

It is worth keeping in mind that unlike the biallelic \opname{} matrix, the
columns of $\mawt{a}$ (or the columns of $\mawt{0}$) are not all
mutually orthogonal.
\end{remark}

\section{\Additive{} \phenoscape{}s}

\begin{define}\label{additiveDef}
A \phenoscape{} $v \in \Csp$ is said to be \emph{\additive{}} if for every
allele $a_i$ at locus $i$ there exists an \emph{additive effect} $\phi_v(a_i,i)$ such
that for all genotypes $\gtype$,

\begin{equation*}
v_{\gtype} = \sum_{i = 1}^{\nl} \phi_v(a_i,i)
\end{equation*}
\end{define}

\begin{thm}
Using the notation from Corollary \ref{orthoDecomp}, we have the following
equivalence:
\begin{equation}\label{additiveSpan}
\text{The set of additive \phenoscape{}s}= S_0 \oplus S_1
\end{equation}
\end{thm}

\begin{proof}
$\subseteq$:
We use the term \emph{unit element for allele $c_k$ at locus $k$} to
denote the
additive \phenoscape{} $u$ with the property
\begin{equation*}
\phi_u(b_j,j) = \delta_{(c_k,k), (b_j,j)}
\end{equation*}
for any allele $b_j$ and locus $j$.
It should be clear from Definition \ref{additiveDef} that
\begin{equation*}
\Span(\text{unit elements}) = \text{set of additive \phenoscape{}s}
\end{equation*}
Let $u \in \Csp$ be the unit element corresponding to the allele-locus
pair $(c_k,k)$, and let $\gtype$ be a genotype with degree $\ge
2$. Then there must be a locus $m \ne k$ where $a_m \ne 0$.

We compute the coefficient of $u$ at $\raisebf{\xgtype{b}}{m}$ using
Corollary
\ref{coeffFinder} and ignoring the factor $1/\na{i}\cdots\na{\nl}$.
\begin{align*}
\sum_{\xgtype{b}}&\mawt{0}_{\gtype,\xgtype{b}} u_{\xgtype{b}}\\
&= \sum_{\xgtype{b}, b_k = c_k}
\prod_{i = 1}^n \wronma{a_i}{b_i}\\
&= \sum_{\brokengtype{b}{m},b_k = c_k}
    \prod_{i \ne
    m}\wronma{a_i}{b_i}\sum_{b_m}\wronma{a_m}{b_m}&&
    \text{(Since $m \ne k$)}\\
&= 0&&\text{(Since $a_m \ne 0$)}
\end{align*}
Thus the coefficients of $u$ in the basis
$\{\raisebf{\gtype}{m}\}_{\gtype}$ are non-zero only for basis elements
with degree $\le 1$, proving the inclusion.

\medskip
\noindent
$\supseteq$:  From Corollary \ref{orthoDecomp}, $S_0$ is spanned by
$\raisebf{0|\cdots|0}{z}$, and from \eqref{mawt0},
\begin{equation*}
\raisebf{0|\cdots|0}{z} = \sum_{\gtype} \vec{\gtype}.
\end{equation*}
This is clearly the additive \phenoscape{} with effect $1/\nl$ for all
alleles.

From Corollary \ref{orthoDecomp}, $S_1$ is spanned by vectors which have
the form\break$\raisebf{0|\cdots|a_j|\cdots|0}{z}$, where $a_j \ne 0$.
From \eqref{mawt0}, we have
\begin{align*}
% Wrote out the top expression explicitly without macros.  I needed the
% alignment mark in the middle to get the rest of the expressions to
% align the way that I wanted them to.
% \raisebf{0|\cdots|a_j|\cdots|0}{z}\\
[0|\cdots|&a_j|\cdots|0]^{\bf \zsymbol}\\
&= \sum_{\xgtype{b}}
\wronma{a_j}{b_j}
\wronma{0}{b_j}^{\nl - 1}
\vec{\xgtype{b}}\\
&= \sum_{\xgtype{b}}
\wronma{a_j}{b_j} \vec{\xgtype{b}}&\text{(Since $\wronma{0}{b_j} = 1$)}\\
&= \sum_{b_1|\cdots|b_{j-1}|0|b_{j+1}|\cdots|b_n} \vec{\xgtype{b}}\\
&\hskip 2em - \sum_{b_1|\cdots|b_{j-1}|a_j|b_{j+1}|\cdots|b_n}
\vec{\xgtype{b}}
\end{align*}
This is an additive \phenoscape{} with effect $1/\nl$ for all alleles at
loci $\ne j$, effect $1/\nl$ for the 0 allele at locus $j$, effect $(1 -
2\nl)/\nl$ for allele $a_j$ at loci $j$, and $(1-\nl)/\nl$ at all other
alleles at locus $j$.

This completes the proof.
\end{proof}

\section{Further examples}

In this \final{} section, we continue in the spirit of Examples
\ref{example1stOrder} and \ref{example0thOrder}
and interpret basis vectors in the form $\raisebf{\gtype}{m}$ or
$\raisebf{\gtype}{z}$ as linear forms and give an interpretation in the
language of interacting effects.

Consider the case with two loci where each has three alleles.
The following are the corresponding matrices for $\mawt{a}$ and
$\mawt{0}$.

\begin{equation}\label{mawtaExampleConcrete}
\mawt{a} =
\begin{blockarray}{rrrrrrrrrr}
& 0|0 & 0|1 & 0|2 & 1|0 & 2|0 & 1|1 & 1|2 & 2|1 & 2|2\\
\begin{block}{r(rrrrrrrrr)}
0|0 & 1 & 1 & 1 & 1 & 1 & 1 & 1 & 1 & 1\\
0|1 & 1 & -2 & 1 & 1 & 1 & -2 & 1 & -2 & 1\\
0|2 & 1 & 1 & -2 & 1 & 1 & 1 & -2 & 1 & -2\\
1|0 & 1 & 1 & 1 & -2 & 1 & -2 & -2 & 1 & 1\\
2|0 & 1 & 1 & 1 & 1 & -2 & 1 & 1 & -2 & -2\\
1|1 & 1 & -2 & 1 & -2 & 1 & 4 & -2 & -2 & 1\\
1|2 & 1 & 1 & -2 & -2 & 1 & -2 & 4 & 1 & -2\\
2|1 & 1 & -2 & 1 & 1 & -2 & -2 & 1 & 4 & -2\\
2|2 & 1 & 1 & -2 & 1 & -2 & 1 & -2 & -2 & 4\\
\end{block}
\end{blockarray}
\end{equation}

\begin{equation}\label{mawt0ExampleConcrete}
\mawt{0} =
\begin{blockarray}{rrrrrrrrrr}
& 0|0 & 0|1 & 0|2 & 1|0 & 2|0 & 1|1 & 1|2 & 2|1 & 2|2\\
\begin{block}{r(rrrrrrrrr)}
0|0 & 1 & 1 & 1 & 1 & 1 & 1 & 1 & 1 & 1\\
0|1 & 1 & -1 & 0 & 1 & 1 & -1 & 0 & -1 & 0\\
0|2 & 1 & 0 & -1 & 1 & 1 & 0 & -1 & 0 & -1\\
1|0 & 1 & 1 & 1 & -1 & 0 & -1 & -1 & 0 & 0\\
2|0 & 1 & 1 & 1 & 0 & -1 & 0 & 0 & -1 & -1\\
1|1 & 1 & -1 & 0 & -1 & 0 & 1 & 0 & 0 & 0\\
1|2 & 1 & 0 & -1 & -1 & 0 & 0 & 1 & 0 & 0\\
2|1 & 1 & -1 & 0 & 0 & -1 & 0 & 0 & 1 & 0\\
2|2 & 1 & 0 & -1 & 0 & -1 & 0 & 0 & 0 & 1\\
\end{block}
\end{blockarray}
\end{equation}

The block diagonality in Theorem \ref{mainTheorem}, including the
orthogonal decomposition with respect to degree asserted in Corollary
\ref{orthoDecomp} is seen when we compute $\mawt{a}^2$ and $\mawt{0}^2$.

\begin{equation*}
\mawt{a}^2 =
\begin{blockarray}{rrrrrrrrrr}
& 0|0 & 0|1 & 0|2 & 1|0 & 2|0 & 1|1 & 1|2 & 2|1 & 2|2\\
\begin{block}{r(rrrrrrrrr)}
0|0 & 9 & 0 & 0 & 0 & 0 & 0 & 0 & 0 & 0\\
0|1 & 0 & 18 & -9 & 0 & 0 & 0 & 0 & 0 & 0\\
0|2 & 0 & -9 & 18 & 0 & 0 & 0 & 0 & 0 & 0\\
1|0 & 0 & 0 & 0 & 18 & -9 & 0 & 0 & 0 & 0\\
2|0 & 0 & 0 & 0 & -9 & 18 & 0 & 0 & 0 & 0\\
1|1 & 0 & 0 & 0 & 0 & 0 & 36 & -18 & -18 & 9\\
1|2 & 0 & 0 & 0 & 0 & 0 & -18 & 36 & 9 & -18\\
2|1 & 0 & 0 & 0 & 0 & 0 & -18 & 9 & 36 & -18\\
2|2 & 0 & 0 & 0 & 0 & 0 & 9 & -18 & -18 & 36\\
\end{block}
\end{blockarray}
\end{equation*}

\begin{equation*}
\mawt{0}^2 =
\begin{blockarray}{rrrrrrrrrr}
& 0|0 & 0|1 & 0|2 & 1|0 & 2|0 & 1|1 & 1|2 & 2|1 & 2|2\\
\begin{block}{r(rrrrrrrrr)}
0|0 & 9 & 0 & 0 & 0 & 0 & 0 & 0 & 0 & 0\\
0|1 & 0 & 6 & 3 & 0 & 0 & 0 & 0 & 0 & 0\\
0|2 & 0 & 3 & 6 & 0 & 0 & 0 & 0 & 0 & 0\\
1|0 & 0 & 0 & 0 & 6 & 3 & 0 & 0 & 0 & 0\\
2|0 & 0 & 0 & 0 & 3 & 6 & 0 & 0 & 0 & 0\\
1|1 & 0 & 0 & 0 & 0 & 0 & 4 & 2 & 2 & 1\\
1|2 & 0 & 0 & 0 & 0 & 0 & 2 & 4 & 1 & 2\\
2|1 & 0 & 0 & 0 & 0 & 0 & 2 & 1 & 4 & 2\\
2|2 & 0 & 0 & 0 & 0 & 0 & 1 & 2 & 2 & 4\\
\end{block}
\end{blockarray}
\end{equation*}

We saw in Examples \ref{example1stOrder} and
\ref{example0thOrder} how some of the basis vectors generated by
$\mawt{a}$ and $\mawt{0}$ gave rise to zeroth and first order
linear forms. Having all the vectors in column form in
\eqref{mawtAExample} and \eqref{mawt0Example} allows us to
compute, for any \phenoscape{} $v$,

\begin{equation}\label{exampleForm2-1}
\innerp{v}{\raisebf{1|1}{\asymbol}}
= v_{00} - 2v_{01} + v_{02} - 2v_{10} + v_{20} + 4v_{11} - 2v_{12} - 2v_{21} + v_{22}
\end{equation}
and
\begin{equation}\label{exampleForm2-2}
\innerp{v}{\raisebf{1|1}{\zsymbol}}
= v_{00} - v_{01} - v_{10} + v_{11}
\end{equation}

If we compare these forms to the analogous form arising from the
biallelic \wt{}, i.e.\
\begin{equation*}
u_{00} - u_{01} - u_{10}+ u_{11}
\end{equation*}
we see that \eqref{exampleForm2-2} is essentially the same.

The form \eqref{exampleForm2-1} is a little trickier to interpret.  We
can start with the difference between $v_{11}$ and the average of its
alternatives which differ at both loci:

\begin{equation}\label{bulk}
v_{11} - \frac{1}{4}(v_{00} + v_{02} + v_{20} + v_{22})
\end{equation}
The averaged first order effects of replacing an allele with allele 1 in
a single
locus (cf.\ Example \ref{example1stOrder}) are given by

\begin{align}
** \rightarrow *1:&\quad\frac{1}{3}(v_{01} + v_{11} + v_{21})
    - \frac{1}{6}(v_{00} + v_{02} + v_{10} + v_{20} + v_{12} + v_{22})
    \label{firstOrderExample1}\\
** \rightarrow 1*:&\quad\frac{1}{3}(v_{10} + v_{11} + v_{12})
    - \frac{1}{6}(v_{00} + v_{01} + v_{02} + v_{20} + v_{21} + v_{22})
    \label{firstOrderExample2}
\end{align}

If we want the ``pure'' second order effect of substituting `11' for one
of `00', `02', `20', or `22', we should subtract
\eqref{firstOrderExample1} and \eqref{firstOrderExample2} from
\eqref{bulk}.  When we do so, we obtain one twelfth of
\eqref{exampleForm2-1}, a useful interpretation of that form.

\bibliographystyle{alpha}

\bibliography{mal_walsh}
\end{document}